\documentclass[twocolumn]{revtex4-2}
\usepackage{amsfonts,amssymb,amsmath}
\usepackage{amsthm}
\usepackage[bookmarks]{hyperref}
\usepackage{physics}
\usepackage{xcolor}
\usepackage{graphicx}
\usepackage{duckuments}
\newcounter{counter}

\newtheorem{proposition}[counter]{Proposition}
\newtheorem{corollary}[counter]{Corollary}
\renewcommand{\epsilon}{\varepsilon}
\newcommand{\pt}[1]{{#1}^{T_B}}
\newcommand{\PPT}{\mathrm{PPT}}
\newcommand{\SN}{\mathrm{SN}}

\begin{document}
\title{Catalytic advantage in asymptotic entanglement manipulation}
\author{Ray Ganardi}
\email{ray@ganardi.xyz}
\affiliation{%
	Scuola Normale Superiore,
	Piazza dei Cavalieri 7,
	56126 Pisa, Italy
}

\begin{abstract}
	Entanglement is a key quantum resource in various quantum protocols, with a rich set of laws governing its manipulation.
	In this context, catalysis refers to the possibility of an auxiliary state that enables a previously forbidden manipulation, while being completely returned at the end.
	While the catalytic setting has been thoroughly examined in the single-copy regime, much less is known in the asymptotically many copy regime.
	In this work, we focus on the entanglement cost of preparing asymptotically many copies of a given state exactly.
	We show that catalysis can significantly lower the exact entanglement cost by constructing an explicit catalytic protocol.
	Additionally, these findings generalize readily to other resource theories, showing a general catalytic advantage in the resource dilution task.
\end{abstract}

\maketitle

Entanglement is a key ingredient in quantum protocols, enabling quantum advantage in various settings ranging from communication, sensing, to computing applications~\cite{Horodecki_2009a}.
A basic prerequisite to utilize entanglement is its generation, which is a common bottleneck in all leading platforms.
For example, generating entangled photon pairs is commonly done through the spontaneous parametric down-conversion process, whose count rates are much lower than other processes in photonic protocols.
Similarly, the error rate of two-qubit gates on trapped ions quantum computers are often orders of magnitude higher than that of one-qubit gates.
The efficient manipulation of entanglement is of paramount importance, which motivates a deeper examination of entanglement manipulation protocols and understand its fundamental limits.

This is the aim of entanglement theory.
The starting point is the framework of local operations and classical communication, also known as LOCC\@.
With these operations, we can prepare any separable state for free.
However, given access to enough entangled states, we can go beyond the LOCC restriction and implement any global operations.
In this way, entangled states acts as a quantum resource, and the manipulation of entanglement reduces to studying transformations between entangled states.

The laws governing entanglement manipulation are complex and exhibit various curious phenomena.
A particularly interesting one is known as catalysis~\cite{Jonathan_1999,Datta_2023a,Lipka-Bartosik_2024a};
there exist states $\rho, \sigma$ such that no LOCC can transform $\rho$ into $\sigma$, but with the addition of a catalyst $\tau$, $\rho \otimes \tau$ can be transformed to $\sigma \otimes \tau$.
While initially defined in the single-copy regime, this catalytic setting was shown to be deeply linked to the asymptotically many-copy regime~\cite{Duan_2005,Ganardi_2024,Müller_2018,Kondra_2021,Shiraishi_2021,Wilming_2021,Lipka-Bartosik_2021}.
In particular, if $\rho$ can be transformed to $\sigma$ asymptotically with at least unit rate, then there exists a catalytic protocol that transforms $\rho$ to $\sigma$.
Catalysis allow us to bridge the gap between the single-copy and asymptotic regime.

Together, these results show that the catalytic setting is at least as powerful as the asymptotic one.
However, it is not clear whether catalysis can be more powerful.
We can formalize the question as follows: can catalysis display some advantage in the asymptotic regime?
Despite the extensive literature on asymptotic protocols, not much is known about catalysis in asymptotic regime.
Previously, only one example of catalytic advantage is known;
in the resource theory of asymmetry, catalysis can increase the distillable coherence from zero to infinity~\cite{Kondra_2024,Shiraishi_2024}.
However, this is due to the peculiar fact that in the resource theory of asymmetry, catalysis can already increase the coherence in a state in the single-copy regime~\cite{Ding_2021}.
Notably, this advantage does not extend to other resource theories, e.g.\ entanglement~\cite{Ganardi_2024,Lami_2024}.

Resolving the question of catalytic advantage in the asymptotic setting is important in two respects:
a negative answer would definitively characterize the power of catalytic transformations.
However, a positive answer will open a new path to obtaining higher rates in key asymptotic tasks.

In this article, we will show a catalytic advantage in the asymptotic setting.
In particular, we will consider the task of exact entanglement dilution, whose figure of merit is the exact entanglement cost.
We construct an example that show that catalysis can lower the exact entanglement cost by giving an explicit catalytic protocol.
On a technical level, we show a close connection between catalytic advantage and non-convexity of the exact entanglement cost.
Furthermore, we show that this phenomenon is not restricted to entanglement theory;
other resource theories also display a catalytic advantage in the resource dilution task.
These results show that even in the asymptotic regime, catalysis can still bring some tangible benefits.

\section{Preliminaries}

In entanglement theory, a Bell pair $\ket{\Phi} = \frac{1}{\sqrt{2}} \pqty{\ket{00} + \ket{11}}$ serves as a standard unit of entanglement.
This is because with enough copies of Bell pairs, we can prepare any state by LOCC, for example by utilizing the teleportation protocol.
This task gives rise to the entanglement measure known as \emph{exact entanglement cost}.
For a given state $\rho$, the exact entanglement cost is defined as follows:
a rate $r$ is called achievable if for any $\delta > 0$, there exists $m, n$, and an LOCC protocol $\Lambda$ such that
\begin{align}
	\Lambda(\Phi^{\otimes m})
		&=
		\rho^{\otimes n},
		\\
		\frac{m}{n}
		&\leq
		r + \delta.
\end{align}
The infimum over all achievable $r$ is known as the exact LOCC entanglement cost $E_c^{0} (\rho)$.

In the catalytic setting, we are given access to an additional system, i.e.\ the catalyst, that must be returned in the same state at the end.
A rate $r$ is called catalytically achievable if for any $\delta > 0$, there exists $m, n$, a catalyst $\tau$, and an LOCC protocol $\Lambda$ such that
\begin{align}
	\Tr_C \Lambda(\Phi^{\otimes m}_S \otimes \tau_C)
		&=
		\rho^{\otimes n}_S,
		\\
		\Tr_S \Lambda(\Phi^{\otimes m}_S \otimes \tau_C)
		&=
		\tau_C,
		\\
		\frac{m}{n}
		&\leq
		r + \delta.
\end{align}
The infimum over all catalytically achievable $r$ is known as the catalytic exact LOCC entanglement cost $E_{c,c}^{0} (\rho)$.

Due to the difficulty of handling LOCC directly, we often study modified versions of entanglement theory.
The theory of PPT entanglement starts with the following simple observation:
a separable state remains positive under partial transpose, i.e.\ it is a PPT state~\cite{Peres_1996}.
This was further developed by Rains~\cite{Rains_1999}, defining the set of PPT operations to obtain an upper bound to distillable entanglement under LOCC\@.
We say that a completely positive trace-preserving map $\Lambda$ is a PPT operation if $T_B \circ \Lambda \circ T_B$ is a completely positive map, where $T_B$ is the partial transpose on the computational basis of system $B$.
It can be shown that this condition is equivalent to:
(1) for any PPT state $\sigma_{A_1 A_2 B_1 B_2}$, the state $\pqty{\mathbf{1}_{A_1 B_1} \otimes \Lambda_{A_2 B_2}} \pqty{\sigma_{A_1 A_2 B_1 B_2}}$ is PPT, and
(2) the Choi state of $\Lambda$ is a PPT state.
We define the exact PPT entanglement cost and the catalytic exact PPT entanglement cost by requiring $\Lambda$ to be a PPT operation instead of LOCC\@.
We can verify that any LOCC protocol is a PPT operation, thus the (catalytic) exact PPT entanglement cost is always a lower bound to the LOCC variant.
While the theory of PPT entanglement was initially defined as an outer approximation of LOCC entanglement, it displays many of the characteristic features of LOCC entanglement.
For example, similarly to LOCC entanglement, PPT entanglement theory is irreversible~\cite{Wang_2017}.

Our primary reason to consider PPT theory instead of LOCC in this article is the following characterization of exact cost for states $\rho$ with positive binegativity $\pt{\abs{\pt{\rho}}} \geq 0$~\cite{Audenaert_2003}:
\begin{align}
	E_c^0 (\rho)
	=
	L_N(\rho),
\end{align}
where $L_N(\rho) = \log_2 \norm{\pt{\rho}}_1$ is the logarithmic negativity~\cite{Plenio_2005}.
The set of states with positive binegativity includes many entangled states that are commonly considered in the literature, such as the Werner states and all two-qubit states.
Furthermore, it is a convex set, closed under tensor products, and it includes the set of PPT states.

\section{Asymptotic catalytic advantage}

Since the discovery of Duan et al.\ that single-shot catalytic transformations are at least as powerful as the asymptotically many-copy regime~\cite{Duan_2005}, an obvious question is whether this advantage persists in the many-copy regime.
To our knowledge, the earliest catalytic advantage related to the asymptotic regime was found in Ref.~\cite{Feng_2006}, where a transformation between two bipartite pure states is possible catalytically in the single-shot regime, but not possible with $n$-to-$n$ multi-copy transformations.
However, this is a very mild advantage that will not show up in transformation rates;
indeed, it is possible for $\rho$ to be transformed to $\sigma$ with unit rate, even though $\rho^{\otimes n}$ cannot be transformed to $\sigma^{\otimes n}$.
This is because when computing transformation rates, any sublinear terms can be ignored.
In fact, the results of Ref.~\cite{Kondra_2021} ensures that there are no catalytic advantage in asymptotic transformations between bipartite pure states.

Our results will show that the situation can change dramatically for transformations between mixed states.
We can see this as a consequence of the existence of different regimes of catalytic transformations in mixed states.
When the final state of the catalyst must be uncorrelated to the system (also called strict catalysis), the additivity of logarithmic negativity prevents us from using Proposition~\ref{proposition:cost-upper-bound} to find examples of catalytic advantage.
In contrast, when we allow some system-catalyst correlation in the final state (i.e. correlated catalysis), Proposition~\ref{proposition:cost-upper-bound} shows examples of catalytic advantage.

Let us start with an upper bound to the catalytic cost.
Suppose we have a state $\rho \in \mathcal{B}(\mathcal{H})$.
We say that a state $\mu \in \mathcal{B}(\mathcal{H}^{\otimes n})$ is an $n$-copy broadcast of $\rho$ if for all $i = 1, \ldots, n$, we have $\Tr_{\overline{i}} \mu = \rho$~\cite{Piani_2009b}, where $\Tr_{\overline{i}}: \mathcal{B}(\mathcal{H}^{\otimes n}) \to \mathcal{B}(\mathcal{H})$ denotes partial trace on all but the $i$th subsystem.
In other words, all the marginals of $\mu$ is in the state $\rho$.
For example, $\rho^{\otimes n}$ is an $n$-copy broadcast of $\rho$.
However, in general there may be correlations between the different copies of $\rho$,
e.g. the state $\ket{\psi} = \sum_i \sqrt{p_i} \ket{ii}$ is a $2$-copy broadcast of $\rho = \sum_i p_i \ketbra{i}$.
We denote the set of all $n$-copy broadcasts of $\rho$ as $\mathfrak{B}_n(\rho)$.

\begin{proposition}\label{proposition:cost-upper-bound}
	For any state $\rho$, we have
	\begin{align}
		E_{c, c}^{0} (\rho) \leq \inf_{\mu \in \mathfrak{B}_2(\rho)} \frac{1}{2} E_c^{0} (\mu).
	\end{align}
\end{proposition}
\begin{proof}
	First, let us show that for any $\mu \in \mathfrak{B}_2(\rho)$, there exists a catalytic protocol that prepares $\rho$ at cost $\frac{1}{2} E_c^{0} (\mu)$.
	Let us fix an arbitrary $\mu \in \mathfrak{B}_2(\rho)$.
	By definition of $E_c^{0}(\mu)$, for any $\delta > 0$, there exist $m, n, \Lambda$ such that
	\begin{align}
		\Lambda(\Phi^{\otimes m})
		&=
		\mu^{\otimes n},
		\\
		\frac{m}{n}
		&\leq
		E_c^{0} (\mu) + \delta.
	\end{align}

	Now, we are ready to construct the catalytic protocol (see Figure~\ref{figure:catalytic-protocol} for an illustration).
	Fix an arbitrary $\delta > 0$.
	First, we prepare the catalyst in the state $\rho^{\otimes n}_{C} $.
	Then, with the system in the state $\Phi^{\otimes m}_{S}$, we apply $\Lambda$ to transform the system state to $\mu^{\otimes n}_{SS'}$.
	Note that $\mu^{\otimes n}_{SS'}$ is composed of $2n$ subsystems.
	Finally, we perform a swap between subsystems $S'$ and $C$, obtaining $\mu^{\otimes n}_{SC} \otimes \rho^{\otimes n}_{S'}$.
	We can verify that the final reduced state of the catalyst is $\rho^{\otimes n}_C$ and the final reduced state of the system is $\rho^{\otimes n}_S \otimes \rho^{\otimes n}_{S'}$.
	This shows that there exists a valid catalytic dilution protocol with rate $\frac{m}{2n}$, i.e.\ $E_{c, c}^{0} (\rho) \leq m/2n \leq \pqty{E_c^{0} (\mu) + \delta}/2$.
	We obtain the claim by taking infimum over all $\delta > 0$ and $\mu \in \mathfrak{B}_2(\rho)$.
\end{proof}

\begin{figure}
	\includegraphics[width=\columnwidth]{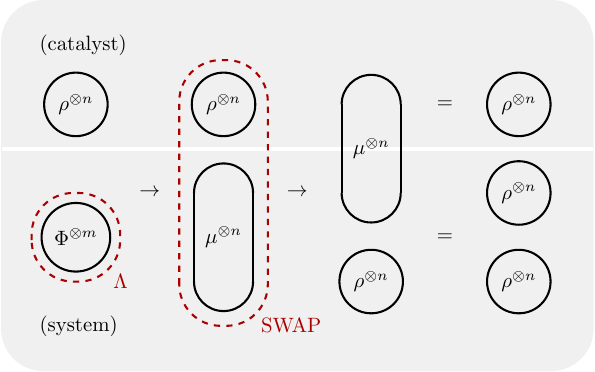}
	\caption{\label{figure:catalytic-protocol}
		An illustration of the catalytic protocol in Proposition~\ref{proposition:cost-upper-bound}.
	}
\end{figure}

With Proposition~\ref{proposition:cost-upper-bound}, it is easy to find examples of a state whose entanglement cost can be lowered by catalysis.
Let us focus on PPT entanglement theory for now, as it will be instructive.
Let us take the following Werner state
\begin{align}
	\rho
	&=
	\frac{1}{2} \Phi_d
	+ \frac{1}{2} \frac{\mathbf{1}}{d^2},
\end{align}
and its $2$-copy broadcast
\begin{align}
	\mu
	&=
	\frac{1}{2} \pqty{%
		\Phi_d \otimes \frac{\mathbf{1}}{d^2}
		+ \frac{\mathbf{1}}{d^2} \otimes \Phi_d
	},
\end{align}
where $\ket{\Phi_d} = \sum_{i < d} \frac{1}{\sqrt{d}} \ket{ii}$.
It is easy to verify that both $\rho$ and $\mu$ has positive binegativity, and thus its exact PPT entanglement cost is given by its logarithmic negativity $L_N(\rho) = \log_2 \norm{\pt{\rho}}_1$~\cite{Audenaert_2003}.
An explicit calculation shows that $L_N(\rho) =  L_N(\mu) = \log_2\pqty{\frac{d^2 + 1}{d}} - 1$, showing that $E_{c, c}^0(\rho) \leq E_{c}^0 (\mu)/2 < E_c^0 (\rho)$.
Thus, we have an example of a state whose catalytic exact PPT entanglement cost is strictly lower than its non-catalytic cost, showing a catalytic advantage in an asymptotic setting.

A closer analysis reveals that this catalytic advantage is a consequence of the remarkable fact that the PPT exact cost is not convex.
To see this, suppose we have a counterexample to midpoint-convexity $\rho = \frac{1}{2} \pqty{\sigma_0 + \sigma_1}$, such that $E_c^0 (\rho) > \pqty{E_c^0(\sigma_0) + E_c^0(\sigma_1)}/2$.
Then, we can verify that $\mu = (\sigma_0 \otimes \sigma_1 + \sigma_1 \otimes \sigma_0)/2$ is a $2$-copy broadcast of $\rho$, and furthermore
\begin{align}
	E_c^0 (\mu)
	\leq E_c^0 (\sigma_0) + E_c^0 (\sigma_1)
	< 2 E_c^0 (\rho).
\end{align}
Applying Proposition~\ref{proposition:cost-upper-bound}, we see that the catalytic cost of $\rho$ is strictly lower than its standard cost.
We have thus proved the following proposition:

\begin{proposition}\label{proposition:convex}
	Suppose $\rho = \frac{1}{2} \pqty{\sigma_0 + \sigma_1}$ is a counterexample to the midpoint-convexity of $E_c^0$, i.e.\ $
	E_c^0 \pqty{\frac{1}{2} \sigma_0 + \frac{1}{2} \sigma_1}
	>
	\frac{1}{2} E_c^0 (\sigma_0)
	+ \frac{1}{2} E_c^0 (\sigma_1)
	$.
	Then $
	E_{c,c}^0 \pqty{\rho}
	< E_c^0 \pqty{\rho}
	$.
\end{proposition}

Incidentally, we have shown that the exact PPT entanglement cost cannot be strongly superadditive, since otherwise we will have $E_c^0 (\rho) \leq E_{c,c}^0 (\rho)$, i.e.\ there will not be any catalytic advantage.
By this, we mean that there exists a state $\rho_{12}$ such that $E_c^0 (\rho_{12}) < E_c^0 (\rho_1) + E_c^0 (\rho_2)$.
\begin{corollary}
	If $E_c^0$ is not midpoint-convex, then it is not strongly superadditive.
\end{corollary}

Note that these arguments do not rely on the fact that we are studying the theory of PPT entanglement.
Let us now turn to the theory of LOCC entanglement.
The exact LOCC cost is given by the regularized log Schmidt number~\cite{Terhal_2000a}:
\begin{align*}
	E_c^0 (\rho)
	&=
	\lim_{n \to \infty}
	\frac{1}{n}
	\log{\SN(\rho^{\otimes n})},
	\\
	\SN(\rho)
	&=
	\inf_{\sum_i p_i \psi_i = \rho} \bqty{%
		\sup_{\psi_i} \rank(\Tr_B \psi_i)
	}.
\end{align*}
On the set of pure states, this is simply the logarithm of the Schmidt rank of the state.
Now, because Schmidt rank only takes integer values, the exact entanglement cost is not continuous on pure states.
Furthermore, for any state $\rho \in \mathcal{B}(\mathcal{H}_{AB})$, we have $E_c^0 (\rho) \leq \log_2 \min \Bqty{\dim{\mathcal{H}_A}, \dim{\mathcal{H}_B}}$, thus it is a bounded function.
These two observations imply that the exact LOCC cost is not midpoint-convex, through the following theorem of Jensen~\cite{Jensen_1906}:
a midpoint-convex function that is bounded on an interval is continuous on that interval.
Applying Proposition~\ref{proposition:convex}, we conclude that the exact LOCC cost must display a catalytic advantage.
This resolves the question on the catalytic behavior of the regularized log Schmidt number (also called asymptotic Schmidt characteristic), dating back to Ref.~\cite{Terhal_2000a}.

\begin{proposition}
	There exists a state $\rho$ such that its catalytic exact LOCC entanglement cost is strictly lower than its exact LOCC entanglement cost,
	i.e.\ $E_{c,c}^0 (\rho) < E_c^0 (\rho)$.
\end{proposition}

We can ask whether we can see the same catalytic advantage in exact distillation.
Here, we start with many copies of a state $\rho$ and aim to obtain as many copies of $\Phi$ as possible, with the optimal rate being the exact distillation rate.
The answer turns out to be no.
Let us start by defining the exact transformation rate $R^0(\rho \to \sigma)$.
A rate $r$ is called achievable when for any $\delta > 0$, there exists $m, n$ and an LOCC protocol $\Lambda$ such that
\begin{align}
	\Lambda(\rho^{\otimes n})
	&=
	\sigma^{\otimes m}
	\\
	\frac{m}{n}
	&\geq r - \delta.
\end{align}
The supremum over all achievable $r$ is the exact transformation rate $R^0(\rho \to \sigma)$.
From this definition, we can see that $E_c^0 (\rho) = R^0 (\Phi \to \rho)^{-1}$.
The exact distillation rate is the dual to this quantity, $E_d^0 (\rho) = R^0 (\rho \to \Phi)$.
We can prove the following analogue of Proposition~\ref{proposition:cost-upper-bound} for distillation:
\begin{align}
	E_{d, c}^0(\rho) \geq \sup_{\mu \in \mathfrak{B}_2(\Phi)} 2 R^0 (\rho \to \mu),
\end{align}
as a consequence of the more general bound $R_{c}^0 (\rho \to \sigma) \geq \sup_{\mu \in \mathfrak{B}_2(\sigma)} 2 R^0(\rho \to \mu)$.
However, due to the purity of $\Phi$, we can see that $\mathfrak{B}_2(\Phi)$ only contains a single element, $\Phi^{\otimes 2}$.
Combined with the fact that $R^0 (\rho \to \Phi^{\otimes 2}) = R^0 (\rho \to \Phi)/2$, we can only get the trivial bound $E_{d, c}^0(\rho) \geq E_{d}^0(\rho)$.
In fact, we can show that this inequality is actually an equality, by showing that the exact distillable entanglement is a strongly superadditive monotone.
Then, standard arguments imply that $E_{d, c}^0(\rho) \leq E_{d}^0(\rho)$.
Thus exact distillable entanglement cannot display any catalytic advantage, mirroring the situation in the approximate regime~\cite{Ganardi_2024,Lami_2024}.

The main reason for the drastic difference in catalytic advantage between cost and distillation is that in a distillation process, the target state is a pure state.
This precludes any correlations in the final state between the system and catalyst, prohibiting the protocol of Proposition~\ref{proposition:cost-upper-bound} to obtain any catalytic advantage.
We also see this phenomenon reflected in the fact that the exact cost for pure states cannot be lowered by catalysis.
To see this, note that for any pure state $\psi$ and any state $\tau$, $E_c^0 (\psi \otimes \tau) = E_c^0 (\psi) + E_c^0 (\tau)$, which can be seen from the explicitly known formula for exact entanglement cost.

\section{General resource theories}
Since our example of a catalytic advantage relies on the non-convexity of exact cost, we can look for similar phenomena in other resource theories.
Let us take the resource theory of thermodynamics~\cite{Janzing_2000,Horodecki_2013,Lostaglio_2019a,Nelly_Ng_2018}.
Here, the analogue of the exact entanglement cost is known as the exact work cost.
For semiclassical states, i.e.\ states that are diagonal in the energy eigenstates, it is given by the max free energy $W_c^0 (\rho) = D_{\max} \pqty{\rho \| \gamma}$~\cite{Horodecki_2013}, which is merely quasi-convex function, \emph{not} convex.
Here, $D_{\max} (\rho \| \sigma) = \inf\Bqty{\log s \,|\, \rho \leq s \sigma}$ is the max-relative entropy~\cite{Datta_2009}.
Thus, we can expect that the exact work cost also displays a catalytic advantage.

To give an explicit example, let us take a two-level system with Gibbs state $\gamma = (1-p) \ketbra{0} + p \ketbra{1}$, with $0 < p < 1/2$.
Then, we can verify that for $0 \leq q \leq 1$, we have
\begin{align}
	W_c^0 \pqty{\ketbra{0}}
	&= \log_2{\frac{1}{1-p}},
	\\
	W_c^0 \pqty{\gamma}
	&= 0,
	\\
	W_c^0 \pqty{(1-q) \ketbra{0} + q \gamma}
	&=
	\log_2{\frac{1-pq}{1-p}},
\end{align}
and thus $
W_c^0 \pqty{\frac{1}{2} \ketbra{0} + \frac{1}{2} \gamma}
>
(W_c^0 \pqty{\ketbra{0}} + W_c^0 \pqty{\gamma})/2
$.
A simple application of Proposition~\ref{proposition:convex} shows that for the state $
\rho = \frac{1}{2} \ketbra{0} + \frac{1}{2} \gamma
$, we have $
W_{c,c}^0 \pqty{\rho}
<
W_c^0 \pqty{\rho}
$.

This exemplifies a common scenario in general resource theories~\cite{Chitambar_2019}, where the resource cost to prepare a particular state $\rho$ exactly is given by $D_{\max}(\rho \| S) = \inf_{\sigma \in S} D_{\max} (\rho \| \sigma)$, where $S$ is the set of free states.
Now, suppose that $\sigma^{*} \in S$ achieves the optimal for $\rho$, i.e.\ $D_{\max} (\rho \| \sigma^{*}) = D_{\max} (\rho \| S)$.
Then we can show that
\begin{align}
\frac{1}{2} D_{\max} (\rho \| S)
+ \frac{1}{2} D_{\max} (\sigma^* \| S)
<
D_{\max} \pqty{\left. \frac{1}{2} \rho + \frac{1}{2} \sigma^{*} \right\| S}
\end{align}
due to the strict concavity of $\log$.
Proposition~\ref{proposition:convex} then implies that there must be a catalytic advantage.

This argument provides a general motif to identifying catalytic effects in the asymptotic regime.
In fact, it explains our first example with PPT entanglement, with an explicit calculation showing that for any Werner state, its logarithmic negativity is equal to $D_{\max} (\rho \| \PPT)$.

\section{Discussions}

We have shown that catalysis can provide an advantage in the asymptotic regime, namely, in the task of exact entanglement dilution.
We showed that the advantage in this task is closely related to the mathematical property of non-convexity of the exact entanglement cost.
This suggests a general approach to finding catalytic advantage in the asymptotic manipulation of other quantum resources.

Several interesting questions remain open.
One interesting question is to find a similar advantage in the approximate setting, where an asymptotically vanishing error is allowed in the transformation.
We can extend Proposition~\ref{proposition:cost-upper-bound} to this setting without much difficulty.
However, it is known that the approximate cost is always convex, prohibiting the application of Proposition~\ref{proposition:convex}.
This is because of the existence of an asymptotic mixing protocol~\cite{arxiv_Vidal_2002}, i.e.\ it is possible to transform $\rho^{\otimes n} \otimes \sigma^{\otimes n}$ to $\pqty{ \frac{\rho + \sigma}{2} }^{\otimes 2n}$ asymptotically, if a vanishing error is allowed.
Thus, more work needs to be done to find an example.

Our example of a catalytic advantage relies on the fact that the exact entanglement cost is non-convex, while the catalytic exact entanglement cost is upper bounded by the convex hull of the exact entanglement cost.
Thus, convexity might be the key distinction between these two measures.
An obvious question is: is the catalytic exact entanglement cost a convex measure?

A common procedure to obtain asymptotic entanglement measures is the notion of regularization;
for a measure $E(\rho)$, the regularized measure is
\begin{align}
	E^{\infty}(\rho)
	&=
	\lim_{n \to \infty}
	\frac{1}{n} E(\rho^{\otimes n}).
\end{align}
We can use the notion of $n$-copy broadcast and the methods in Ref.~\cite{Ganardi_2024} to obtain a converse to our results, namely the following lower bound on catalytic exact entanglement cost
\begin{align}
	E_{c,c}^0 (\rho)
	\geq
	\lim_{n \to \infty}
	\inf_{\mu \in \mathfrak{B}_n (\rho)}
	\frac{1}{n} E_{c}^0 (\mu).
\end{align}
Here, the right hand side is akin to a regularized quantity, apart from the infimum over all $n$-copy broadcasts.
This is the notion of \emph{broadcast regularization}~\cite{Piani_2009b}.
Comparing to the upper bound in Proposition~\ref{proposition:cost-upper-bound}, this suggests that the exact catalytic cost is closely connected to broadcast-regularized quantities.
We have seen that for logarithmic negativity and max-relative entropy of resource, broadcast regularization produces a different monotone compared to the usual regularization.
It would be interesting to see whether this also holds for more well-behaved measures such as the relative entropy of entanglement.

More generally, our results paint a more optimistic picture of the real cost of entanglement manipulation tasks.
In addition to significantly increasing the power of single-shot transformations, we have discovered that asymptotic transformations can also benefit from catalysis.
Indeed, for the state $\rho = \frac{1}{2} \Phi_d + \frac{1}{2} \frac{\mathbf{1}}{d^2}$, catalysis lowers the entanglement cost by at least a factor of two, compared to the standard asymptotic setting.
This suggests that we might obtain significantly higher transformation rates in experiments if we allow for alternative settings such as catalysis.

\begin{acknowledgments}
RG thanks Samrat Sen for the discussion on exact LOCC cost.
RG acknowledges financial support from the European Union (ERC StG ETQO, Grant Agreement no.\ 101165230). Views and opinions expressed are however those of the author(s) only and do not necessarily reflect those of the European Union or the European Research Council. Neither the European Union nor the granting authority can be held responsible for them.
\end{acknowledgments}

\bibliography{refs.bib}

\begin{thebibliography}{31}%
\makeatletter
\providecommand \@ifxundefined [1]{%
 \@ifx{#1\undefined}
}%
\providecommand \@ifnum [1]{%
 \ifnum #1\expandafter \@firstoftwo
 \else \expandafter \@secondoftwo
 \fi
}%
\providecommand \@ifx [1]{%
 \ifx #1\expandafter \@firstoftwo
 \else \expandafter \@secondoftwo
 \fi
}%
\providecommand \natexlab [1]{#1}%
\providecommand \enquote  [1]{``#1''}%
\providecommand \bibnamefont  [1]{#1}%
\providecommand \bibfnamefont [1]{#1}%
\providecommand \citenamefont [1]{#1}%
\providecommand \href@noop [0]{\@secondoftwo}%
\providecommand \href [0]{\begingroup \@sanitize@url \@href}%
\providecommand \@href[1]{\@@startlink{#1}\@@href}%
\providecommand \@@href[1]{\endgroup#1\@@endlink}%
\providecommand \@sanitize@url [0]{\catcode `\\12\catcode `\$12\catcode
  `\&12\catcode `\#12\catcode `\^12\catcode `\_12\catcode `\%12\relax}%
\providecommand \@@startlink[1]{}%
\providecommand \@@endlink[0]{}%
\providecommand \url  [0]{\begingroup\@sanitize@url \@url }%
\providecommand \@url [1]{\endgroup\@href {#1}{\urlprefix }}%
\providecommand \urlprefix  [0]{URL }%
\providecommand \Eprint [0]{\href }%
\providecommand \doibase [0]{https://doi.org/}%
\providecommand \selectlanguage [0]{\@gobble}%
\providecommand \bibinfo  [0]{\@secondoftwo}%
\providecommand \bibfield  [0]{\@secondoftwo}%
\providecommand \translation [1]{[#1]}%
\providecommand \BibitemOpen [0]{}%
\providecommand \bibitemStop [0]{}%
\providecommand \bibitemNoStop [0]{.\EOS\space}%
\providecommand \EOS [0]{\spacefactor3000\relax}%
\providecommand \BibitemShut  [1]{\csname bibitem#1\endcsname}%
\let\auto@bib@innerbib\@empty
%</preamble>
\bibitem [{\citenamefont {Horodecki}\ \emph {et~al.}(2009)\citenamefont
  {Horodecki}, \citenamefont {Horodecki}, \citenamefont {Horodecki},\ and\
  \citenamefont {Horodecki}}]{Horodecki_2009a}%
  \BibitemOpen
  \bibfield  {author} {\bibinfo {author} {\bibfnamefont {R.}~\bibnamefont
  {Horodecki}}, \bibinfo {author} {\bibfnamefont {P.}~\bibnamefont
  {Horodecki}}, \bibinfo {author} {\bibfnamefont {M.}~\bibnamefont
  {Horodecki}},\ and\ \bibinfo {author} {\bibfnamefont {K.}~\bibnamefont
  {Horodecki}},\ }\bibfield  {title} {\bibinfo {title} {Quantum entanglement},\
  }\href {https://doi.org/10.1103/RevModPhys.81.865} {\bibfield  {journal}
  {\bibinfo  {journal} {Rev. Mod. Phys.}\ }\textbf {\bibinfo {volume} {81}},\
  \bibinfo {pages} {865} (\bibinfo {year} {2009})}\BibitemShut {NoStop}%
\bibitem [{\citenamefont {Jonathan}\ and\ \citenamefont
  {Plenio}(1999)}]{Jonathan_1999}%
  \BibitemOpen
  \bibfield  {author} {\bibinfo {author} {\bibfnamefont {D.}~\bibnamefont
  {Jonathan}}\ and\ \bibinfo {author} {\bibfnamefont {M.~B.}\ \bibnamefont
  {Plenio}},\ }\bibfield  {title} {\bibinfo {title} {Entanglement-assisted
  local manipulation of pure quantum states},\ }\href
  {https://doi.org/10.1103/physrevlett.83.3566} {\bibfield  {journal} {\bibinfo
   {journal} {Physical Review Letters}\ }\textbf {\bibinfo {volume} {83}},\
  \bibinfo {pages} {3566} (\bibinfo {year} {1999})}\BibitemShut {NoStop}%
\bibitem [{\citenamefont {Datta}\ \emph {et~al.}(2023)\citenamefont {Datta},
  \citenamefont {Varun~Kondra}, \citenamefont {Miller},\ and\ \citenamefont
  {Streltsov}}]{Datta_2023a}%
  \BibitemOpen
  \bibfield  {author} {\bibinfo {author} {\bibfnamefont {C.}~\bibnamefont
  {Datta}}, \bibinfo {author} {\bibfnamefont {T.}~\bibnamefont {Varun~Kondra}},
  \bibinfo {author} {\bibfnamefont {M.}~\bibnamefont {Miller}},\ and\ \bibinfo
  {author} {\bibfnamefont {A.}~\bibnamefont {Streltsov}},\ }\bibfield  {title}
  {\bibinfo {title} {Catalysis of entanglement and other quantum resources},\
  }\href {https://doi.org/10.1088/1361-6633/acfbec} {\bibfield  {journal}
  {\bibinfo  {journal} {Reports on Progress in Physics}\ }\textbf {\bibinfo
  {volume} {86}},\ \bibinfo {pages} {116002} (\bibinfo {year}
  {2023})}\BibitemShut {NoStop}%
\bibitem [{\citenamefont {Lipka-Bartosik}\ \emph {et~al.}(2024)\citenamefont
  {Lipka-Bartosik}, \citenamefont {Wilming},\ and\ \citenamefont
  {Ng}}]{Lipka-Bartosik_2024a}%
  \BibitemOpen
  \bibfield  {author} {\bibinfo {author} {\bibfnamefont {P.}~\bibnamefont
  {Lipka-Bartosik}}, \bibinfo {author} {\bibfnamefont {H.}~\bibnamefont
  {Wilming}},\ and\ \bibinfo {author} {\bibfnamefont {N.~H.~Y.}\ \bibnamefont
  {Ng}},\ }\bibfield  {title} {\bibinfo {title} {Catalysis in quantum
  information theory},\ }\href {https://doi.org/10.1103/revmodphys.96.025005}
  {\bibfield  {journal} {\bibinfo  {journal} {Reviews of Modern Physics}\
  }\textbf {\bibinfo {volume} {96}},\ \bibinfo {pages} {025005} (\bibinfo
  {year} {2024})}\BibitemShut {NoStop}%
\bibitem [{\citenamefont {Duan}\ \emph {et~al.}(2005)\citenamefont {Duan},
  \citenamefont {Feng}, \citenamefont {Li},\ and\ \citenamefont
  {Ying}}]{Duan_2005}%
  \BibitemOpen
  \bibfield  {author} {\bibinfo {author} {\bibfnamefont {R.}~\bibnamefont
  {Duan}}, \bibinfo {author} {\bibfnamefont {Y.}~\bibnamefont {Feng}}, \bibinfo
  {author} {\bibfnamefont {X.}~\bibnamefont {Li}},\ and\ \bibinfo {author}
  {\bibfnamefont {M.}~\bibnamefont {Ying}},\ }\bibfield  {title} {\bibinfo
  {title} {Multiple-copy entanglement transformation and entanglement
  catalysis},\ }\href {https://doi.org/10.1103/physreva.71.042319} {\bibfield
  {journal} {\bibinfo  {journal} {Physical Review A}\ }\textbf {\bibinfo
  {volume} {71}},\ \bibinfo {pages} {042319} (\bibinfo {year}
  {2005})}\BibitemShut {NoStop}%
\bibitem [{\citenamefont {Ganardi}\ \emph {et~al.}(2024)\citenamefont
  {Ganardi}, \citenamefont {Kondra},\ and\ \citenamefont
  {Streltsov}}]{Ganardi_2024}%
  \BibitemOpen
  \bibfield  {author} {\bibinfo {author} {\bibfnamefont {R.}~\bibnamefont
  {Ganardi}}, \bibinfo {author} {\bibfnamefont {T.~V.}\ \bibnamefont
  {Kondra}},\ and\ \bibinfo {author} {\bibfnamefont {A.}~\bibnamefont
  {Streltsov}},\ }\bibfield  {title} {\bibinfo {title} {Catalytic and
  asymptotic equivalence for quantum entanglement},\ }\href
  {https://doi.org/10.1103/physrevlett.133.250201} {\bibfield  {journal}
  {\bibinfo  {journal} {Physical Review Letters}\ }\textbf {\bibinfo {volume}
  {133}},\ \bibinfo {pages} {250201} (\bibinfo {year} {2024})}\BibitemShut
  {NoStop}%
\bibitem [{\citenamefont {Müller}(2018)}]{Müller_2018}%
  \BibitemOpen
  \bibfield  {author} {\bibinfo {author} {\bibfnamefont {M.~P.}\ \bibnamefont
  {Müller}},\ }\bibfield  {title} {\bibinfo {title} {Correlating thermal
  machines and the second law at the nanoscale},\ }\href
  {https://doi.org/10.1103/physrevx.8.041051} {\bibfield  {journal} {\bibinfo
  {journal} {Physical Review X}\ }\textbf {\bibinfo {volume} {8}},\ \bibinfo
  {pages} {041051} (\bibinfo {year} {2018})}\BibitemShut {NoStop}%
\bibitem [{\citenamefont {Kondra}\ \emph {et~al.}(2021)\citenamefont {Kondra},
  \citenamefont {Datta},\ and\ \citenamefont {Streltsov}}]{Kondra_2021}%
  \BibitemOpen
  \bibfield  {author} {\bibinfo {author} {\bibfnamefont {T.~V.}\ \bibnamefont
  {Kondra}}, \bibinfo {author} {\bibfnamefont {C.}~\bibnamefont {Datta}},\ and\
  \bibinfo {author} {\bibfnamefont {A.}~\bibnamefont {Streltsov}},\ }\bibfield
  {title} {\bibinfo {title} {Catalytic transformations of pure entangled
  states},\ }\href {https://doi.org/10.1103/physrevlett.127.150503} {\bibfield
  {journal} {\bibinfo  {journal} {Physical Review Letters}\ }\textbf {\bibinfo
  {volume} {127}},\ \bibinfo {pages} {150503} (\bibinfo {year}
  {2021})}\BibitemShut {NoStop}%
\bibitem [{\citenamefont {Shiraishi}\ and\ \citenamefont
  {Sagawa}(2021)}]{Shiraishi_2021}%
  \BibitemOpen
  \bibfield  {author} {\bibinfo {author} {\bibfnamefont {N.}~\bibnamefont
  {Shiraishi}}\ and\ \bibinfo {author} {\bibfnamefont {T.}~\bibnamefont
  {Sagawa}},\ }\bibfield  {title} {\bibinfo {title} {Quantum thermodynamics of
  correlated-catalytic state conversion at small scale},\ }\href
  {https://doi.org/10.1103/physrevlett.126.150502} {\bibfield  {journal}
  {\bibinfo  {journal} {Physical Review Letters}\ }\textbf {\bibinfo {volume}
  {126}},\ \bibinfo {pages} {150502} (\bibinfo {year} {2021})}\BibitemShut
  {NoStop}%
\bibitem [{\citenamefont {Wilming}(2021)}]{Wilming_2021}%
  \BibitemOpen
  \bibfield  {author} {\bibinfo {author} {\bibfnamefont {H.}~\bibnamefont
  {Wilming}},\ }\bibfield  {title} {\bibinfo {title} {Entropy and reversible
  catalysis},\ }\href {https://doi.org/10.1103/physrevlett.127.260402}
  {\bibfield  {journal} {\bibinfo  {journal} {Physical Review Letters}\
  }\textbf {\bibinfo {volume} {127}},\ \bibinfo {pages} {260402} (\bibinfo
  {year} {2021})}\BibitemShut {NoStop}%
\bibitem [{\citenamefont {Lipka-Bartosik}\ \emph {et~al.}(2021)\citenamefont
  {Lipka-Bartosik}, \citenamefont {Mazurek},\ and\ \citenamefont
  {Horodecki}}]{Lipka-Bartosik_2021}%
  \BibitemOpen
  \bibfield  {author} {\bibinfo {author} {\bibfnamefont {P.}~\bibnamefont
  {Lipka-Bartosik}}, \bibinfo {author} {\bibfnamefont {P.}~\bibnamefont
  {Mazurek}},\ and\ \bibinfo {author} {\bibfnamefont {M.}~\bibnamefont
  {Horodecki}},\ }\bibfield  {title} {\bibinfo {title} {Second law of
  thermodynamics for batteries with vacuum state},\ }\href
  {https://doi.org/10.22331/q-2021-03-10-408} {\bibfield  {journal} {\bibinfo
  {journal} {Quantum}\ }\textbf {\bibinfo {volume} {5}},\ \bibinfo {pages}
  {408} (\bibinfo {year} {2021})}\BibitemShut {NoStop}%
\bibitem [{\citenamefont {Kondra}\ \emph {et~al.}(2024)\citenamefont {Kondra},
  \citenamefont {Ganardi},\ and\ \citenamefont {Streltsov}}]{Kondra_2024}%
  \BibitemOpen
  \bibfield  {author} {\bibinfo {author} {\bibfnamefont {T.~V.}\ \bibnamefont
  {Kondra}}, \bibinfo {author} {\bibfnamefont {R.}~\bibnamefont {Ganardi}},\
  and\ \bibinfo {author} {\bibfnamefont {A.}~\bibnamefont {Streltsov}},\
  }\bibfield  {title} {\bibinfo {title} {Coherence manipulation in asymmetry
  and thermodynamics},\ }\href {https://doi.org/10.1103/physrevlett.132.200201}
  {\bibfield  {journal} {\bibinfo  {journal} {Physical Review Letters}\
  }\textbf {\bibinfo {volume} {132}},\ \bibinfo {pages} {200201} (\bibinfo
  {year} {2024})}\BibitemShut {NoStop}%
\bibitem [{\citenamefont {Shiraishi}\ and\ \citenamefont
  {Takagi}(2024)}]{Shiraishi_2024}%
  \BibitemOpen
  \bibfield  {author} {\bibinfo {author} {\bibfnamefont {N.}~\bibnamefont
  {Shiraishi}}\ and\ \bibinfo {author} {\bibfnamefont {R.}~\bibnamefont
  {Takagi}},\ }\bibfield  {title} {\bibinfo {title} {Arbitrary amplification of
  quantum coherence in asymptotic and catalytic transformation},\ }\href
  {https://doi.org/10.1103/physrevlett.132.180202} {\bibfield  {journal}
  {\bibinfo  {journal} {Physical Review Letters}\ }\textbf {\bibinfo {volume}
  {132}},\ \bibinfo {pages} {180202} (\bibinfo {year} {2024})}\BibitemShut
  {NoStop}%
\bibitem [{\citenamefont {Ding}\ \emph {et~al.}(2021)\citenamefont {Ding},
  \citenamefont {Hu},\ and\ \citenamefont {Fan}}]{Ding_2021}%
  \BibitemOpen
  \bibfield  {author} {\bibinfo {author} {\bibfnamefont {F.}~\bibnamefont
  {Ding}}, \bibinfo {author} {\bibfnamefont {X.}~\bibnamefont {Hu}},\ and\
  \bibinfo {author} {\bibfnamefont {H.}~\bibnamefont {Fan}},\ }\bibfield
  {title} {\bibinfo {title} {Amplifying asymmetry with correlating catalysts},\
  }\href {https://doi.org/10.1103/physreva.103.022403} {\bibfield  {journal}
  {\bibinfo  {journal} {Physical Review A}\ }\textbf {\bibinfo {volume}
  {103}},\ \bibinfo {pages} {022403} (\bibinfo {year} {2021})}\BibitemShut
  {NoStop}%
\bibitem [{\citenamefont {Lami}\ \emph {et~al.}(2024)\citenamefont {Lami},
  \citenamefont {Regula},\ and\ \citenamefont {Streltsov}}]{Lami_2024}%
  \BibitemOpen
  \bibfield  {author} {\bibinfo {author} {\bibfnamefont {L.}~\bibnamefont
  {Lami}}, \bibinfo {author} {\bibfnamefont {B.}~\bibnamefont {Regula}},\ and\
  \bibinfo {author} {\bibfnamefont {A.}~\bibnamefont {Streltsov}},\ }\bibfield
  {title} {\bibinfo {title} {No-go theorem for entanglement distillation using
  catalysis},\ }\href {https://doi.org/10.1103/physreva.109.l050401} {\bibfield
   {journal} {\bibinfo  {journal} {Physical Review A}\ }\textbf {\bibinfo
  {volume} {109}},\ \bibinfo {pages} {L050401} (\bibinfo {year}
  {2024})}\BibitemShut {NoStop}%
\bibitem [{\citenamefont {Peres}(1996)}]{Peres_1996}%
  \BibitemOpen
  \bibfield  {author} {\bibinfo {author} {\bibfnamefont {A.}~\bibnamefont
  {Peres}},\ }\bibfield  {title} {\bibinfo {title} {Separability criterion for
  density matrices},\ }\href {https://doi.org/10.1103/physrevlett.77.1413}
  {\bibfield  {journal} {\bibinfo  {journal} {Physical Review Letters}\
  }\textbf {\bibinfo {volume} {77}},\ \bibinfo {pages} {1413} (\bibinfo {year}
  {1996})}\BibitemShut {NoStop}%
\bibitem [{\citenamefont {Rains}(1999)}]{Rains_1999}%
  \BibitemOpen
  \bibfield  {author} {\bibinfo {author} {\bibfnamefont {E.~M.}\ \bibnamefont
  {Rains}},\ }\bibfield  {title} {\bibinfo {title} {Bound on distillable
  entanglement},\ }\href {https://doi.org/10.1103/physreva.60.179} {\bibfield
  {journal} {\bibinfo  {journal} {Physical Review A}\ }\textbf {\bibinfo
  {volume} {60}},\ \bibinfo {pages} {179} (\bibinfo {year} {1999})}\BibitemShut
  {NoStop}%
\bibitem [{\citenamefont {Wang}\ and\ \citenamefont {Duan}(2017)}]{Wang_2017}%
  \BibitemOpen
  \bibfield  {author} {\bibinfo {author} {\bibfnamefont {X.}~\bibnamefont
  {Wang}}\ and\ \bibinfo {author} {\bibfnamefont {R.}~\bibnamefont {Duan}},\
  }\bibfield  {title} {\bibinfo {title} {Irreversibility of asymptotic
  entanglement manipulation under quantum operations completely preserving
  positivity of partial transpose},\ }\href
  {https://doi.org/10.1103/physrevlett.119.180506} {\bibfield  {journal}
  {\bibinfo  {journal} {Physical Review Letters}\ }\textbf {\bibinfo {volume}
  {119}},\ \bibinfo {pages} {180506} (\bibinfo {year} {2017})}\BibitemShut
  {NoStop}%
\bibitem [{\citenamefont {Audenaert}\ \emph {et~al.}(2003)\citenamefont
  {Audenaert}, \citenamefont {Plenio},\ and\ \citenamefont
  {Eisert}}]{Audenaert_2003}%
  \BibitemOpen
  \bibfield  {author} {\bibinfo {author} {\bibfnamefont {K.}~\bibnamefont
  {Audenaert}}, \bibinfo {author} {\bibfnamefont {M.~B.}\ \bibnamefont
  {Plenio}},\ and\ \bibinfo {author} {\bibfnamefont {J.}~\bibnamefont
  {Eisert}},\ }\bibfield  {title} {\bibinfo {title} {Entanglement cost under
  positive-partial-transpose-preserving operations},\ }\href
  {https://doi.org/10.1103/physrevlett.90.027901} {\bibfield  {journal}
  {\bibinfo  {journal} {Physical Review Letters}\ }\textbf {\bibinfo {volume}
  {90}},\ \bibinfo {pages} {027901} (\bibinfo {year} {2003})}\BibitemShut
  {NoStop}%
\bibitem [{\citenamefont {Plenio}(2005)}]{Plenio_2005}%
  \BibitemOpen
  \bibfield  {author} {\bibinfo {author} {\bibfnamefont {M.~B.}\ \bibnamefont
  {Plenio}},\ }\bibfield  {title} {\bibinfo {title} {Logarithmic negativity: A
  full entanglement monotone that is not convex},\ }\href
  {https://doi.org/10.1103/physrevlett.95.090503} {\bibfield  {journal}
  {\bibinfo  {journal} {Physical Review Letters}\ }\textbf {\bibinfo {volume}
  {95}},\ \bibinfo {pages} {090503} (\bibinfo {year} {2005})}\BibitemShut
  {NoStop}%
\bibitem [{\citenamefont {Feng}\ \emph {et~al.}(2006)\citenamefont {Feng},
  \citenamefont {Duan},\ and\ \citenamefont {Ying}}]{Feng_2006}%
  \BibitemOpen
  \bibfield  {author} {\bibinfo {author} {\bibfnamefont {Y.}~\bibnamefont
  {Feng}}, \bibinfo {author} {\bibfnamefont {R.}~\bibnamefont {Duan}},\ and\
  \bibinfo {author} {\bibfnamefont {M.}~\bibnamefont {Ying}},\ }\bibfield
  {title} {\bibinfo {title} {Relation between catalyst-assisted transformation
  and multiple-copy transformation for bipartite pure states},\ }\href
  {https://doi.org/10.1103/physreva.74.042312} {\bibfield  {journal} {\bibinfo
  {journal} {Physical Review A}\ }\textbf {\bibinfo {volume} {74}},\ \bibinfo
  {pages} {042312} (\bibinfo {year} {2006})}\BibitemShut {NoStop}%
\bibitem [{\citenamefont {Piani}\ \emph {et~al.}(2009)\citenamefont {Piani},
  \citenamefont {Christandl}, \citenamefont {Mora},\ and\ \citenamefont
  {Horodecki}}]{Piani_2009b}%
  \BibitemOpen
  \bibfield  {author} {\bibinfo {author} {\bibfnamefont {M.}~\bibnamefont
  {Piani}}, \bibinfo {author} {\bibfnamefont {M.}~\bibnamefont {Christandl}},
  \bibinfo {author} {\bibfnamefont {C.~E.}\ \bibnamefont {Mora}},\ and\
  \bibinfo {author} {\bibfnamefont {P.}~\bibnamefont {Horodecki}},\ }\bibfield
  {title} {\bibinfo {title} {Broadcast copies reveal the quantumness of
  correlations},\ }\href {https://doi.org/10.1103/physrevlett.102.250503}
  {\bibfield  {journal} {\bibinfo  {journal} {Physical Review Letters}\
  }\textbf {\bibinfo {volume} {102}},\ \bibinfo {pages} {250503} (\bibinfo
  {year} {2009})}\BibitemShut {NoStop}%
\bibitem [{\citenamefont {Terhal}\ and\ \citenamefont
  {Horodecki}(2000)}]{Terhal_2000a}%
  \BibitemOpen
  \bibfield  {author} {\bibinfo {author} {\bibfnamefont {B.~M.}\ \bibnamefont
  {Terhal}}\ and\ \bibinfo {author} {\bibfnamefont {P.}~\bibnamefont
  {Horodecki}},\ }\bibfield  {title} {\bibinfo {title} {Schmidt number for
  density matrices},\ }\href {https://doi.org/10.1103/physreva.61.040301}
  {\bibfield  {journal} {\bibinfo  {journal} {Physical Review A}\ }\textbf
  {\bibinfo {volume} {61}},\ \bibinfo {pages} {040301} (\bibinfo {year}
  {2000})}\BibitemShut {NoStop}%
\bibitem [{\citenamefont {Jensen}(1906)}]{Jensen_1906}%
  \BibitemOpen
  \bibfield  {author} {\bibinfo {author} {\bibfnamefont {J.~L. W.~V.}\
  \bibnamefont {Jensen}},\ }\bibfield  {title} {\bibinfo {title} {Sur les
  fonctions convexes et les inégalités entre les valeurs moyennes},\ }\href
  {https://doi.org/10.1007/bf02418571} {\bibfield  {journal} {\bibinfo
  {journal} {Acta Mathematica}\ }\textbf {\bibinfo {volume} {30}},\ \bibinfo
  {pages} {175–193} (\bibinfo {year} {1906})}\BibitemShut {NoStop}%
\bibitem [{\citenamefont {Janzing}\ \emph {et~al.}(2000)\citenamefont
  {Janzing}, \citenamefont {Wocjan}, \citenamefont {Zeier}, \citenamefont
  {Geiss},\ and\ \citenamefont {Beth}}]{Janzing_2000}%
  \BibitemOpen
  \bibfield  {author} {\bibinfo {author} {\bibfnamefont {D.}~\bibnamefont
  {Janzing}}, \bibinfo {author} {\bibfnamefont {P.}~\bibnamefont {Wocjan}},
  \bibinfo {author} {\bibfnamefont {R.}~\bibnamefont {Zeier}}, \bibinfo
  {author} {\bibfnamefont {R.}~\bibnamefont {Geiss}},\ and\ \bibinfo {author}
  {\bibfnamefont {T.}~\bibnamefont {Beth}},\ }\bibfield  {title} {\bibinfo
  {title} {Thermodynamic cost of reliability and low temperatures: Tightening
  landauer's principle and the second law},\ }\href
  {https://doi.org/10.1023/a:1026422630734} {\bibfield  {journal} {\bibinfo
  {journal} {International Journal of Theoretical Physics}\ }\textbf {\bibinfo
  {volume} {39}},\ \bibinfo {pages} {2717} (\bibinfo {year}
  {2000})}\BibitemShut {NoStop}%
\bibitem [{\citenamefont {Horodecki}\ and\ \citenamefont
  {Oppenheim}(2013)}]{Horodecki_2013}%
  \BibitemOpen
  \bibfield  {author} {\bibinfo {author} {\bibfnamefont {M.}~\bibnamefont
  {Horodecki}}\ and\ \bibinfo {author} {\bibfnamefont {J.}~\bibnamefont
  {Oppenheim}},\ }\bibfield  {title} {\bibinfo {title} {Fundamental limitations
  for quantum and nanoscale thermodynamics},\ }\href
  {https://doi.org/10.1038/ncomms3059} {\bibfield  {journal} {\bibinfo
  {journal} {Nature Communications}\ }\textbf {\bibinfo {volume} {4}},\
  \bibinfo {pages} {2059} (\bibinfo {year} {2013})}\BibitemShut {NoStop}%
\bibitem [{\citenamefont {Lostaglio}(2019)}]{Lostaglio_2019a}%
  \BibitemOpen
  \bibfield  {author} {\bibinfo {author} {\bibfnamefont {M.}~\bibnamefont
  {Lostaglio}},\ }\bibfield  {title} {\bibinfo {title} {An introductory review
  of the resource theory approach to thermodynamics},\ }\href
  {https://doi.org/10.1088/1361-6633/ab46e5} {\bibfield  {journal} {\bibinfo
  {journal} {Reports on Progress in Physics}\ }\textbf {\bibinfo {volume}
  {82}},\ \bibinfo {pages} {114001} (\bibinfo {year} {2019})}\BibitemShut
  {NoStop}%
\bibitem [{\citenamefont {Ng}\ and\ \citenamefont
  {Woods}(2018)}]{Nelly_Ng_2018}%
  \BibitemOpen
  \bibfield  {author} {\bibinfo {author} {\bibfnamefont {N.~H.~Y.}\
  \bibnamefont {Ng}}\ and\ \bibinfo {author} {\bibfnamefont {M.~P.}\
  \bibnamefont {Woods}},\ }\bibinfo {title} {Resource theory of quantum
  thermodynamics: Thermal operations and second laws},\ in\ \href
  {https://doi.org/10.1007/978-3-319-99046-0_26} {\emph {\bibinfo {booktitle}
  {Thermodynamics in the Quantum Regime}}}\ (\bibinfo  {publisher} {Springer
  International Publishing},\ \bibinfo {year} {2018})\ p.\ \bibinfo {pages}
  {625–650}\BibitemShut {NoStop}%
\bibitem [{\citenamefont {Datta}(2009)}]{Datta_2009}%
  \BibitemOpen
  \bibfield  {author} {\bibinfo {author} {\bibfnamefont {N.}~\bibnamefont
  {Datta}},\ }\bibfield  {title} {\bibinfo {title} {Min- and max-relative
  entropies and a new entanglement monotone},\ }\href
  {https://doi.org/10.1109/tit.2009.2018325} {\bibfield  {journal} {\bibinfo
  {journal} {IEEE Transactions on Information Theory}\ }\textbf {\bibinfo
  {volume} {55}},\ \bibinfo {pages} {2816} (\bibinfo {year}
  {2009})}\BibitemShut {NoStop}%
\bibitem [{\citenamefont {Chitambar}\ and\ \citenamefont
  {Gour}(2019)}]{Chitambar_2019}%
  \BibitemOpen
  \bibfield  {author} {\bibinfo {author} {\bibfnamefont {E.}~\bibnamefont
  {Chitambar}}\ and\ \bibinfo {author} {\bibfnamefont {G.}~\bibnamefont
  {Gour}},\ }\bibfield  {title} {\bibinfo {title} {Quantum resource theories},\
  }\href {https://doi.org/10.1103/revmodphys.91.025001} {\bibfield  {journal}
  {\bibinfo  {journal} {Reviews of Modern Physics}\ }\textbf {\bibinfo {volume}
  {91}},\ \bibinfo {pages} {025001} (\bibinfo {year} {2019})}\BibitemShut
  {NoStop}%
\bibitem [{\citenamefont {Vidal}(2002)}]{arxiv_Vidal_2002}%
  \BibitemOpen
  \bibfield  {author} {\bibinfo {author} {\bibfnamefont {G.}~\bibnamefont
  {Vidal}},\ }\href {https://arxiv.org/abs/quant-ph/0203107} {\bibinfo {title}
  {On the continuity of asymptotic measures of entanglement}} (\bibinfo {year}
  {2002}),\ \Eprint {https://arxiv.org/abs/quant-ph/0203107}
  {arXiv:quant-ph/0203107 [quant-ph]} \BibitemShut {NoStop}%
\end{thebibliography}%

\end{document}